\documentclass{article}
\usepackage[utf8]{inputenc}
\usepackage{appendix}
\usepackage{amsmath, amssymb, amsthm}
\usepackage{booktabs}
\usepackage{graphicx}
\usepackage{color}

\usepackage{geometry}
\usepackage{setspace}
\usepackage{subcaption}
\usepackage[english]{babel}
\usepackage[style=authoryear,natbib=true, sorting=nyt]{biblatex}
\usepackage{csquotes}

\newtheorem{theorem}{Theorem}
\newtheorem{assumption}{Assumption}
\newcommand{\norm}[1]{\left\lVert#1\right\rVert}

\newcommand{\Ber}{\text{Bernoulli}}
\newcommand{\logit}{\text{logit}}

\newcommand{\E}{\mathbb{E}}

\title{Efficient estimation of longitudinal treatment effects using difference-in-differences and machine learning}
\author{Nicholas Illenberger, Iván D\'iaz, Audrey Renson\\[2ex]
\small Department of Population Health, New York University Grossman School of Medicine, New York, United States\\[2ex]
\small Contact: nicholas.illenberger@nyulangone.org}
\date{
    \today
}

\begin{document}

\maketitle

\begin{abstract}
Difference-in-differences is based on a parallel trends assumption, which states that changes over time in average potential outcomes are independent of treatment assignment, possibly conditional on covariates. With time-varying treatments, parallel trends assumptions can identify many types of parameters, but most work has focused on group-time average treatment effects and similar parameters conditional on the treatment trajectory. This paper focuses instead on identification and estimation of the intervention-specific mean – the mean potential outcome had everyone been exposed to a proposed intervention – which may be directly policy-relevant in some settings. Previous estimators for this parameter under parallel trends have relied on correctly-specified parametric models, which may be difficult to guarantee in applications. We develop multiply-robust and efficient estimators of the intervention-specific mean based on the efficient influence function, and derive conditions under which data-adaptive machine learning methods can be used to relax modelling assumptions. Our approach allows the parallel trends assumption to be conditional on the history of time-varying covariates, thus allowing for adjustment for time-varying covariates possibly impacted by prior treatments. Simulation results support the use of the proposed methods at modest sample sizes. As an example, we estimate the effect of a hypothetical federal minimum wage increase on self-rated health in the US.
\end{abstract}

\section{Introduction}
Difference-in-differences (DID) methods are popular in policy analysis because they can identify effects in the presence of unmeasured confounding, which is often expected in comparisons across jurisdictions. DID relies on parallel trends assumptions, which state that time trends in average potential outcomes are independent of treatment assignment (possibly conditional on covariates). This is in contrast to much of the causal inference literature, where identification is based on exchangeability assumptions; these state that levels (rather than trends) of potential outcomes are independent of treatment assignment (possibly conditional on covariates). Parallel trends can (but does not necessarily) hold in the presence of unmeasured common causes of the outcome and treatment (i.e., unmeasured confounding) whereas exchangeability cannot. For example, when treatment only changes once (i.e. is time-invariant), parallel trends holds if and only if there is a constant unmeasured confounding effect pre- and post-treatment \citep{sofer2016negative}. The simluation data-generating mechanism in Section \ref{sec:sim} of this paper provides an example of how parallel trends may hold under unmeasured confounding in a more complicated time-varying scenario.
\par Until recently, most methodological analyses of DID have considered only point treatments \citep{abadie2005semiparametric, heckman1997matching}, but realistic settings often have time-varying treatments, making identification and estimation more complex. Time-varying treatments are usually accomodated in the applied literature via two-way fixed effects (TWFE) models, which rely on strong parametric assumptions and/or may make illogical comparisons, resulting in potentially large bias \citep{goodman2021difference}. More recently, methods have been developed that relax these assumptions by considering DID for time-varying treatments in a non-parametric or semi-parametric framework \citep{callaway2021difference, de2022difference, renson2022identifying, shahn2022structural}. These methods are increasingly used in practice and can produce markedly different results than TWFE \citep{goin2023comparing}. 
\par In DID with time-invariant binary treatments, there is typically only one parameter of interest, the average treatment effect in the treated (ATT). However, with time-varying and/or many-valued treatments, there is not necessarily an obvious generalization of the ATT since there may be many observed treatment trajectories; therefore, many different causal parameters can potentially be of interest. The recent DID literature has focused primarily on parameters built on the ATT concept, such as the group-time average treatment effects \citep{callaway2021difference} or cohort average treatment effects \citep{sun2021estimating}; like the ATT, these parameters condition on an observed treatment trajectory and apply mainly to binary treatments. A different parameter which may also be of interest and can be identified in time-varying DID settings is the intervention-specific mean, or the the average potential outcome for the whole population under a specified treatment plan \citep{renson2022identifying}. Intervention-specific means have been the major focus in the literature on  longitudinal treatment effects under sequential exchangeability assumptions, which propose that potential outcomes are independent of a time-varying treatment given the history of covariates and treatments. Intervention-specific means have received less attention in the DID literature, but may be more policy-relevant in some settings. For example, \cite{renson2022identifying} used a DID framework to estimate counterfactual US mortality rates had all US states remained longer under pandemic stay-at-home orders, such a parameter may be directly applicable to pandemic planning decisions at the federal level. 
\par One challenge with time-varying treatments is that there are likely to be time-varying covariates, which may be necessary to condition on in order for the parallel trends assumption to hold. Time-varying covariates in the parallel trends assumption have received limited attention in the DID literature \citep{de2022difference, caetano2022difference, renson2022identifying, shahn2022structural}. When the parallel trends assumption depends on time-varying covariates, the standard approach in applied literature has been TWFE models with covariates entered as regressors. When analyzing these models, it is generally assumed (in addition to the other strong parametric assumptions) that covariates are not affected by past treatment \citep[e.g.][]{zeldow2021confounding,goodman2021difference}; such covariates are termed ``bad controls'' in this literature \citep{angrist2009mostly}. Analyses of DID models that relax the TWFE assumptions have generally also ruled out feedback between treatment and covariates \citep{caetano2022difference, ghanem2022selection}.  However, such treatment-covariate feedback is likely to be common in many DID applications; \cite{callaway2023policy} and \cite{renson2022identifying} both provide examples.
\par Previously developed estimators in time-varying DID settings have generally required that at least one of either the propensity score or outcome regression are correctly-specified parametric models  \citep{callaway2021difference, renson2022identifying}. Such models may be high-dimensional and difficult to specify correctly in practice, especially when estimation depends on the history of time-varying covariates. Estimators that allow for appropriate inference when data-adaptive models are used to estimate nuisance parameters have been developed primarily for settings where identification is based on exchangeability \citep[e.g.][]{diaz2023nonparametric}. Such an approach has been extended to estimating the ATT in a time-invariant DID setting \citep{chang2020double}. The extension of this approach to a time-varying DID setting has, to our knowledge, only been considered in the context of g-estimation of structural nested models \cite{shahn2022structural}. Though this approach allows estimation of the intervention-specific mean and other marginal parameters, the procedure is quite complex and its main strength is in studying time-varying effect heterogeneity.  
\par Building on the identification results of \cite{renson2022identifying}, we develop estimators for the intervention-specific mean under parallel trends based on a nonparametric statistical model. Our estimators are based on the efficient influence function for the identified parameter, and thus are nonparametrically asymptotically efficient, multiply robust, and can accomodate data-adaptive machine learning to estimate the propensity scores and/or outcome regression functions. Our approach depends on a parallel trends assumption that conditions on the history of time-varying covariates up to each treatment time point, and does not require an assumption that covariates are unaffected by prior treatment. Though our focus is on longitudinal settings with repeated observations among units, our results could be extended to repeated cross-sectional data as well. The remainder of this paper is organized as follows. Section 2 presents identification results and estimators based on the efficient influence function, Section 3 presents a simulation study exploring finite-sample properties of the proposed estimators, Section 4 presents a worked example studying the effect of minimum wages on self-rated health, and Section 5 concludes.
\section{Methods}
\subsection{Data}
Suppose we collect data on $i=1,...,n$ participants at $t=0,..., \tau$ discrete timepoints. For participant $i$ at timepoint $t$, we observe a discrete exposure variable $A_{it}$, a set of time-varying covariates $X_{it}$, and the outcome of interest $Y_{it}$. The full collection of observed data takes the form of $n$ i.i.d. copies of $O_i = (X_{i0}, A_{i0}, Y_{i0}, X_{i1}, ..., X_{i\tau}, A_{i\tau}, Y_{i\tau})$. Let $\overline{A}_t = (A_0, A_1,...,A_t)$ denote the treatment history up to and including time $t$, and let $\overline A\equiv\overline A_{\tau}$ denote the entire history. We define $\overline{X}_t$ and $\overline{Y}_t$ similarly. For convenience, we define $Z_t=\{\emptyset\}$ when $t<0$ for any random variable $Z$.   We assume that data come from a staggered discontinuation design \citep{renson2022identifying} such that $A_0=a_0^*$ for all units, where $a_0^*$ is the value that $A_0$ takes under the counterfactual treatment trajectory to which the target causal effect refers (defined in the next section). In other words, we assume all units begin follow-up with treatment value consistent with the counterfactual regime of interest, but may later discontinue. Note that this design includes as a special case staggered adoption designs, often the focus of time-varying DID studies \citep{callaway2021difference}, where $A_t\in\{0,1\}$, $A_0=0$ for everyone, and $A_t=1$ implies $A_{t+1}=1$ with probability 1, but also includes settings where treatments switch on and off and/or are non-binary.

\subsection{Target causal effects}
We will work within the potential outcomes framework, wherein $Y_t(\overline a_t)$ represents the value of $Y_t$ that would be observed if an individual's exposure history had been set by intervention to $\overline a_t$. Note that this potential outcome is only observed for (at most) those participants whose observed treatment pattern follows $\overline A_t=\overline a_t$. 

We will focus on causal quantities of the form $\mu_t = \E[Y_t(\overline a^*)]$, where $\overline a^*$ represents a counterfactual trajectory of interest. Thus, the quantity $\mu_t$ is an ``intervention-specific mean''; i.e., an average potential outcome, had everyone's exposure trajectory $\overline A$ been set by intervention to $\overline a^*$. Typically, counterfactual outcomes are compared to observed outcomes in some way, so that the target causal quantity is actually $\E[Y_t - Y_t(\overline a^*)]$. However, for such comparisons, estimation and inference are typically straightforward once these are achieved for the counterfactual component $\mu_t$; hence, the focus of this article is the latter.   Throughout, we assume there is only one counterfactual trajectory $\overline a^*$ of interest; the methods presented in this article also apply for multiple trajectories, but assuming parallel trends for multiple trajectories places restrictions on treatment effect homogeneity which may not be desirable \citep{renson2022identifying, callaway2021continuous}. 

\subsection{Identification}
Because potential outcomes under a given regime are unobserved, certain assumptions are required to connect the distribution of these potential outcomes to that of the observable data. In our context, identification of the causal effect of $A$ on $Y$ will rely on a generalization of the parallel trends assumption to a longitudinal setting. Specifically, we require the following identification assumptions:
\begin{assumption}[No interference] \label{asn:noint}
$Y_{it}(\overline a_{i}^*, \overline a_{i'}^*)=Y_{it}(\overline a_i)^*$ for $i\neq i'$
\end{assumption}
\begin{assumption}[Treatment version irrelevance] \label{asn:tvi}
If $\overline A_{it} = \overline a^*_t$, then $Y_{it} = Y_{it}(\overline a^*_t)$ 
\end{assumption}
\begin{assumption}[No anticipation] \label{asn:noant}
$Y_{it}(\overline a^*_s)=Y_{it}(\overline a_t^*)$ for $s>t$
\end{assumption}
\begin{assumption}[Parallel trends]\label{asn:pt}
    For $t \in \{1,2,...,\tau\}$ and $k\in \{0, 1,..., t-1\}$, assume there is a $\overline  W_k\subseteq(\overline X_k, \overline Y_{k-2})$ such that
    \[
    E[Y_t(\overline{a}^*) - Y_{t-1}(\overline{a}^*)|\overline W_k, \overline{A}_{k-1} = \overline{a}_{k-1}^*] = E[Y_t(\overline{a}^*) - Y_{t-1}(\overline{a}^*)|\overline W_k, \overline{A}_{k} = \overline{a}_{k}^*]
    \]
    where $\overline W_{k-1} \subseteq \overline W_k$.
\end{assumption}
\begin{assumption}[Positivity]\label{asn:pos}
If $f(w_t|\overline{w}_{t-1}, \overline{a}_{t-1}^*) > 0$, then $f(\overline{a}_t^*|\overline{w}_t, \overline{a}_{t-1}^*) > 0$ for $\overline{w}_t \in \overline{\mathcal{W}}_t$, $t\in\{0,1,...,\tau\}$
\end{assumption}
Assumptions \ref{asn:noint}-\ref{asn:noant} are standard in the longitudinal DID literature; we briefly summarize these here. No interference requires that outcomes of one unit are not affected by the treatment a separate unit receives (i.e. the potential outcome only depends on that unit's observed treatment). Treatment version irrelevance requires that exposures are defined and measured with sufficient granularity that, if a given unit is observed under the treatment pattern $\overline{a}^*$, then their observed outcome is equal to their potential outcome under an intervention to set $\overline A=\overline a^*$. Assumptions \ref{asn:noint} and \ref{asn:tvi} are sometimes combined as the Stable Unit Treatment Value Assumption. No anticipation requires that treatments observed at a time $s>t$ after outcomes $Y_{it}$ did not have the ability to impact that prior outcome; for example, this may be violated if units altered their behavior in anticipation of planned policy changes. 
\par The parallel trends assumption (Assumption \ref{asn:pt}) requires that there is some subset of the covariate history (up to time $k$) and the outcome history (up to time $k-2$), conditional on which the average changes in potential outcomes on the additive scale over successive periods are equal across values of the current treatment. (Note that if $Y_{k-1}$ were included, then parallel trends would reduce to a sequential exchangeability assumption and DID methods would not be necessary.) Assumption \ref{asn:pt} is identical to the parallel trends assumption previously considered in the context of the intervention-specific mean \citep{renson2022identifying}. It differs from parallel trends assumptions considered in other popular longitudinal DID methods in that it conditions on the history of time-varying covariates; for example, the methods of \cite{callaway2021difference} condition only on baseline covariates, those of \cite{de2022difference} condition on contemporaneous covariates, and those of \cite{caetano2022difference} condition on the entire covariate history (i.e. including covariates measured later than the outcome). Importantly, we require no assumption that covariates are unaffected by treatment; thus, our methods can account for treatment-covariate feedback. Finally, positivity requires that there be positive probability of receiving the proposed treatment pattern $a_t^*$, for individuals with treatment history $\overline{a}_{t-1}^*$ and at all levels of the covariate/outcome history needed for parallel trends. 
\par Following \cite{renson2022identifying}, under Assumptions \ref{asn:noint}-\ref{asn:pt}, we have:
\begin{align*}
    \mu_t &= \E[Y_0(\overline a^*)] + \sum_{k=1}^t \E[Y_k(\overline a^*)-Y_{k-1}(\overline a^*)]\\
    &= \E[Y_0(\overline a^*)] +  \sum_{k=1}^t \E(\E[Y_k(\overline a^*)-Y_{k-1}(\overline a^*)|\overline W_1, A_1])\\
     &= \E[Y_0(\overline a^*)] +  \sum_{k=1}^t \E(\E[Y_k(\overline a^*)-Y_{k-1}(\overline a^*)|\overline W_1, A_1=a_1^*])\\
     &\vdots \\
      &= \E[Y_0(\overline a^*)] +  \sum_{k=1}^t \E(\E[\cdots \E\{Y_k(\overline a^*)-Y_{k-1}(\overline a^*)|\overline W_k,\overline A_k\}\cdots|\overline W_1, A_1=a_1^*])\\
      &= \E[Y_0(\overline a^*)] +  \sum_{k=1}^t \E(\E[\cdots \E\{Y_k(\overline a^*)-Y_{k-1}(\overline a^*)|\overline W_k,\overline A_k=\overline a_k^*\}\cdots|\overline W_1, A_1=a_1^*])\\
      &= \E[Y_0(a_0^*)] +  \sum_{k=1}^t \E(\E[\cdots \E\{Y_k(\overline a^*_k)-Y_{k-1}(\overline a^*_k)|\overline W_k,\overline A_k=\overline a_k^*\}\cdots|\overline W_1, A_1=a_1^*])\\
      &= \E[Y_0] +  \sum_{k=1}^t \E(\E[\cdots \E\{Y_k-Y_{k-1}|\overline W_k,\overline A_k=\overline a_k^*\}\cdots|\overline W_1, A_1=a_1^*])\\
\end{align*}
The first equality of the above proof holds by telescoping, the second and fourth by law of iterated expectation, the third and fifth by parallel trends, the sixth by no anticipation, and the seventh by no interference and treatment version irrelevance because $Y_k(\overline a) - Y_{k-1}(\overline a)=Y_k-Y_{k-1}$ in the event $\overline A_k =\overline a_k$. Note that in the last line above, covariates $\overline W_k$ and treatments $\overline A_k$ may be measured later than outcomes $Y_{k-1}$; thus, in estimators we consider below, outcomes may be regressed on future variables, which may appear counter-intuitive. This concern may be addressed by noting that differences $Y_k-Y_{k-1}$ are only determined in the future of $\overline W_k$ and $\overline A_k$.
\par In order to simplify notation in our definitions of estimators, we additionally define:
\begin{align*}
     \E[Y_0] &+  \sum_{k=1}^t \E(\E[\cdots \E\{Y_k-Y_{k-1}|\overline W_k,\overline A_k=\overline a_k^*\}\cdots|\overline W_1, A_1=a_1^*])\\
    &= \E[Y_0] + \sum_{k=1}^t (\phi_{k,k} - \phi_{k-1, k} )\\
    &= \psi_t
\end{align*}
with $\psi_t=\E[Y_0] + \sum_{k=1}^t (\phi_{k,k} - \phi_{k-1, k} )$ and $\phi_{j,k} = \E\{ \E[ \cdots \E(Y_j\mid \overline W_k,\overline A_k=\overline a_k^*)\cdots|\overline W_1, A_1=a_1^*]\}$. Note that $\phi_{j,k}$ is the time-varying g-formula \citep{robins1986new}, which identifies causal effects of the type $\mu_t$ under sequential exchangeability. Here, $\phi_{j,k}$ has no causal interpretation in itself, but instead forms the building blocks of the parameter $\psi_t$ (which equals $\mu_t$ under our assumptions). However, established methods to estimate $\phi_{j,k}$ based on its efficient influence function form the basis of our proposed estimators.
\subsection{Estimators}

For the rest of this paper we assume that Assumptions \ref{asn:noint}-\ref{asn:pos} hold, so that the intervention-specific mean is identified as $\mu_t=\psi_t$, and focus attention on estimating the quantity $\psi_t.$ Before defining the efficient influence function and resulting estimator, we will introduce some additional notation. Let $g_m(\overline W_m)=\prod_{k=1}^m f(A_k=a_k^* |\overline W_k, \overline A_{k-1}=\overline a_{k-1}^*)$ denote the cumulative propensity score up to time $m$. Define $Q^{j, k, k+1} \equiv Y_j$, and let $Q^{j,k,m}$ be iteratively defined as
\[Q^{j,k,m}(\overline{W}_m) = E[Q^{j, k, m+1}(\overline{W}_{m+1})|\overline{A}_m = \overline{a}_m^*, \overline{W}_m].\]
It can be seen that $\phi_{j,k} = E[Q^{j,k,1}(\overline{W}_1)].$ In what follows, we use the notation $P_n X$ to denote the empirical average of $X$, $n^{-1}\sum_{i=1}^n X_i$. Whenever there is no ambiguity, we drop the argument $\overline{W}_m$ from $g_m(\overline W_m)$ and $Q^{j,k,m}(\overline{W}_m)$ for conciseness. For example, we let $\widehat{Q}_{i}^{j,k,m}$ denote an estimate of $Q^{j,k,m}(\overline{W}_m)$ for observation $i$. For a function $f$ of $\overline{W}_m$, we let $||f||^2 = \int [f(\overline{w}_m)]^2dP(\overline{w}_m)$.

The following theorem defines the efficient influence function (EIF) for the parameter $\psi_t$, upon which multiply robust and efficient data-adaptive estimators will be based. 
\begin{theorem}\label{thm:if}
The Efficient Influence Function for $\psi_t$ is given by
\begin{align*}
    IF(\psi_t)&= IF(\E[Y_0]) + \sum_{k=1}^t \big \{ IF(\phi_{k,k}) - IF(\phi_{k-1,k})\big \}
\end{align*}
where $IF(\E[Y_0]) = Y_0 - \E[Y_0]$ and $IF(\phi_{j,k}) = \sum_{m=1}^k  \frac{I(\overline{A}_{m} = \overline{a}_{m}^*)}{g_m( \overline W_m)}(Q^{j,k,m+1} - Q^{j,k,m}) + Q^{j,k,1} - \phi_{j,k}.$
\end{theorem}
Given the EIF a ``one-step'' estimator for $\psi_t$ is given by
\begin{equation}\label{eqn:full-estimator}
    \widehat{\psi}_t = P_n \Bigg\{ Y_0 + \sum_{k=1}^t \bigg(\widetilde{\phi}_{k,k}(\widehat{\Theta}) - \widetilde{\phi}_{k-1,k}(\widehat{\Theta})\bigg)\Bigg\}
\end{equation}
where:
\begin{equation*}
    \widetilde{\phi}_{jki}(\widehat{\Theta}) =  \sum_{m=1}^k \frac{I(\overline{A}_{mi} = \overline{a}^*_m)}{\widehat{g}_{mi}}(\widehat{Q}_{i}^{j,k,m+1} - \widehat{Q}_{i}^{j,k,m}) + \widehat{Q}_{i}^{j,k,1}
\end{equation*}
where $\widehat{Q}_i^{j,k,m}$ and $\widehat g_{mi}=\widehat g_m(\overline W_{mi})$ denote estimators of the referenced quantities, and $\widehat{\Theta}=\{ \widehat{Q}_i^{j,k,m}, \widehat g_m : m=1,...,k; k=1,...,t\}$ denotes the full collection of nuisance function estimators. In the next section we will describe asymptotic behavior of this estimator, including multiple robustness, efficiency, and convergence rates under difference choices of nuisance function estimators.

\subsection{Asymptotic Results}

\par This section describes the asymptotic distribution of the one-step estimator $\widehat \psi_t$ of the intervention-specific mean. This discussion includes (1) inferential techniques based on the asymptotic normality of the estimator, (2) conditions under which the estimator is unbiased, and (3) description of a sample splitting procedure that allows for efficient estimation under flexible data-adaptive estimators of the nuisance functions.
\begin{theorem}\label{thm:asymp}
Assume the following conditions:
\begin{itemize}
    \item[(C1)] The estimated nuisance functions, $\hat\Theta$, are in a Donsker class, and are such that $||IF[\phi_{j,k}(\hat\Theta)] - IF[\phi_{j,k}]||\xrightarrow{p} 0$, and 
    \item[(C2)] For some constant $\epsilon > 0$, $P[\epsilon < \widehat{g}_{k}(\overline W_m)] = P[\epsilon < g_{k}(\overline W_m)]=1$, for all $k$ and $\overline{a}^*$,
\end{itemize}
where $\xrightarrow{p}$ denotes convergence in probability. Let 
\[R_{2,n}=\sum_{m = 1}^t \norm{g_m - \widehat{g}_m}\sum_{k = m}^t \bigg( \norm{Q^{k,k,m} - \widehat{Q}^{k,k,m}}  + \norm{Q^{k-1,k,m} - \widehat{Q}^{k-1,k,m}}\bigg)\]
Under (C1) and (C2), we have:
\[\widehat\psi_t - \psi_t =(P_n - P)IF(\psi_t) + o_p(n^{-1/2})+O_p(R_{2,n}) \]
\end{theorem}

\begin{proof}
    By Lemma 1 of \cite{luedtke2017sequential} we have 
    \[\widetilde{\phi}_{jk} - \phi_{jk} = (P_n - P)IF(\phi_{jk}) + (P_n-P)(\widehat{IF}(\phi_{jk}) - IF(\phi_{jk})) + Rem_{jk},\]
    where $\widehat{IF}(\phi_{jk})$ denotes the IF of Theorem~\ref{thm:if} with $\Theta$ replaced by $\widehat\Theta$, and where 
    \[Rem_{jk} = \sum_{m=1}^k \int c_m(\overline w_m) \{g_m(\overline w_m) - \widehat{g}_m(\overline w_m)\}\{Q^{j,k,m}(\overline w_m) - \widehat{Q}^{j,k,m}(\overline w_m)\}dP(\overline w_m),\]
    where $c_m(\overline W_m)$ is a constant given in the original Lemma, which is bounded in probability according to C2. Under C1, Lemma 19.24 of \cite{van2000asymptotic} shows that the term $(P_n-P)(\widehat{IF}(\phi_{jk}) - IF(\phi_{jk}))$ is $o_P(n^{-1/2})$. Likewise, using the Cauchy-Schwarz inequality we have
    \[|Rem_{jk}| \leq \norm{g_m - \widehat{g}_m}\norm{Q^{j,k,m} - \widehat{Q}^{j,k,m}}.\]
    Applying the above to 
   \[\psi_t=\E[Y_0] + \sum_{k=1}^t (\phi_{k,k} - \phi_{k-1, k} )
     ;\quad \widetilde\psi_t=P_n Y_0 + \sum_{k=1}^t (\widetilde\phi_{k,k} - \widetilde\phi_{k-1, k} )
     \]
     together with Theorem \ref{thm:if} yields the result. 
\end{proof}
Theorem \ref{thm:asymp} describes the asymptotic behavior of our proposed estimator. The first error component illustrates the asymptotic bias of our estimator, as well as the importance of using doubly robust estimating equations such as the EIF when using machine learning. First, note that $R_{2,n}$ is defined essentially by sums of product of errors in the estimation of the nuisance parameters $g_m$ and $Q^{k,k,m}$ as well as $Q^{k-1,k,m}$. Machine learning regression algorithms are specifically designed to make each of these errors small. Therefore, the product of each pair of errors is expected to be smaller than each of its components. In comparison, if these nuisance regressions are estimated with parametric models, large regression errors from model misspecificaiton are more likely. Therefore, one of these errors could be large, which would make the product larger than each of its components \citep{robins2007comment}, leading to large bias. 

Estimators constructed using the efficient influence function as above are often called doubly-robust because they yield consistent estimators even in the case of misspecification of one of the models under consideration. According to the above theorem, the estimator $\widehat\psi_t$ is consistent if $R_{2.n}=o_p(1)$. This holds if, for each time point $m\leq t$, we have either $g_m$ is consistently estimated, or the sequence $S_m=\{(Q^{k,k,m},Q^{k-1,k,m}):m\leq k\leq t\}$ is consistenly estimated, where we note that consistent estimation of $S_m$ requires consistent estimation of $S_{m+1}$. Assume that we start with a consistent outcome regression estimator at time point $t$, but then at time $m\leq t$ we have an inconsistent outcome regression. Then we would require that the propensity score is estimated consistently at all times prior to and including $m$, since in this case $S_j:j\leq m$ cannot be consistently estimated anymore. Since we have $t+1$ possible values for $m$, this means that we have $t+1$ chances of consistency. Importantly, this also implies that using an adaptation of the sequential doubly robust estimators of \cite{luedtke2017sequential} and \cite{rotnitzky2017multiply} would not yield estimators that are more robust than a simple one-step estimator.

The result of the theorem also allows us to study the weak convergence of the estimator. In particular, if $R_{2,n}=o_p(n^{-1/2})$, then it will be the case that 
\[
\sqrt{n}(\widehat\psi_t - \psi_t ) \xrightarrow{d} N[0, \norm{ IF(\psi_t) }^2]
\]
where $\xrightarrow{d}$ denotes convergence in distribution. In this case, an estimator of the asymptotic variance of $\widehat \psi_t$ is given by:
\begin{equation*}
    \widehat{V}(\widehat\psi_t) = \frac{1}{n}\sum_{i=1}^n \left\{ Y_{0i} + \sum_{k=1}^t \bigg(\widetilde{\phi}_{k,k,i}(\widehat{\Theta}) - \widetilde{\phi}_{k-1,k,i}(\widehat{\Theta})\bigg)   - \widehat \psi_t \right\}^2.
\end{equation*}
A sufficient condition for the second order term of the theorem to be $o_p(n^{-1/2})$ is that all nuisance parameters are estimated consistently at rates $n^{-1/4}$. These rates are achievable by some machine learning algorithms, under assumptions.  

At this point it is important to clarify a semantic note on the use of the term ``double robustness''. In this paper, as in most manuscripts that use machine learning for estimation of causal effects, we use the term double robustness to refer to the property that the estimator is consistent under inconsistent estimation of some of the nuisance estimators, as stated above. Some of the literature on causal inference, especially the literature focused on parametric models, uses the term double robustness to refer to the property that the estimators are $\sqrt{n}$-consistent and asymptotically normal under misspecification of some of the nuisance models. While our proposed estimator would also enjoy asymptotic normality under misspecification of some models if $\hat\Theta$ is estimated within pre-specified parametric models, in this paper we are focused on addressing model misspecification bias and therefore focus on results for machine learning nuisance estimation, using the term double-robustness to refer to consistency of the estimator under model misspecification.

\subsection{Sample-Split Estimator}
\par The asymptotic results for the proposed estimator relies on the condition that the estimated nuisance functions, $\hat\Theta$, are within the Donsker class. This condition controls how close functions of the estimated nuisance functions can get to their limiting versions (e.g. the difference between $(P_n - P)\tilde{\phi}_{jk}(\widehat{\Theta})$ and $(P_n - P)\tilde{\phi}_{jk}(\Theta)$) \citep{kennedy2016semiparametric}. This condition can be restrictive if estimators with unbounded variation norm are used, as these are generally not Donsker \citep[e.g.,][]{diaz2019statistical}. This condition can be removed by implementing sample-split versions of the estimators. Partition the index set $\{1,...,n\}$ into $M$ mutually exclusive groups $V_1,...,V_M$ such that $\bigcup_{m=1}^M V_m =\{1,...,n\}$ and $V_m \bigcap V_{m^{'}} = \emptyset$ for $m \ne m^{'}$. For each validation set, $m$, let $\mathcal{T}_m = \{1,...,n \}\setminus\mathcal{V}_m$ denote the corresponding training set. Let $m(i)$ denote the index of the validation set containing unit $i$ (e.g. $m(i) = s$ if $i \in \mathcal{V}_s$). Finally, let $\widehat{\Theta}_m = (\widehat{Q}_m, \widehat{g}_m)$ denote the estimated nuisance functions obtained using training data $\mathcal{T}_m$.
\par The sample-split estimator of the desired causal estimand is obtained by replacing the full-sample estimated nuisance functions in Equation \ref{eqn:full-estimator} with their training-sample estimated counterparts for each unit.

\begin{equation*}
    \widehat{\psi}_{t,ss} = P_n \Bigg\{ Y_{0i} + \sum_{k=1}^t \bigg(\widetilde{\phi}_{k,k}(\widehat{\Theta}_{m(i)}) - \widetilde{\phi}_{k-1,k}(\widehat{\Theta}_{m(i)})\bigg)\Bigg\}
\end{equation*}
Note that within each validation set, $V_m$, the estimated nuisance functions $\widehat{\Theta}$ are fixed. This fact allows us to apply the empirical process properties that our asymptotic results are based upon without invoking Donsker class conditions.
\par Because estimated treatment effects depend on the random partitioning of the data, we may repeat the procedure over different partitions and combine estimates of the treatment effect and variance. Let $\widehat{\psi}_{t,ss,k}$ for $k=1,...,K$ and $ \widehat{V}(\widehat{\psi}_{t,ss,k})$ denote the estimated treatment effect and variance for the $k$-th repetition. The combined treatment effect can be estimated as the median estimate across the $K$ partitions of the data, $\widetilde{\psi}_{t,ss} = \text{median}_k(\widehat{\psi}_{t,ss,k})$ \citep{zivich2021machine}. As opposed to the mean, the median is more robust to outliers across the partitioned estimates \citep{chernozhukov2018double}. The estimated variance of this estimator is given by $\widehat{V}(\widetilde{\psi}_{t,ss}) = \text{median}_k(\widehat{V}(\widehat{\psi}_{t,ss,k}) + (\widehat{\psi}_{t,ss,k} - \widetilde{\psi}_{t,ss})^2)$. This estimator incorporates the uncertainty arising from the random partitioning of the data.

\section{Simulation study}\label{sec:sim}
We generated 500 datasets each at sample sizes $n=1,000$, $5,000$, and $20,000$ according to the following structural model for $t=0,1,2$:
\begin{align*}
    U &\sim N(0,1)\\
    W_{1t} &\sim N(\gamma_{01t} + \gamma_{11t}W_{1,t-1} + \gamma_{21t}A_{t-1}, 0.1)\\
    W_{2t} &\sim N(\gamma_{02t} + \gamma_{12t}W_{1t} + \gamma_{22t}W_{2, t-1} + \gamma_{32t}A_{t-1}, 0.1)\\
    W_{3t} &\sim N(\gamma_{03t} + \gamma_{13t}W_{1t} + \gamma_{23t}W_{2t} + 
    \gamma_{33t}W_{3, t-1}+\gamma_{43t}A_{t-1}, 0.1)\\
    A_t &\sim \Ber\{ (1-A_{t-1})\logit(\alpha_{0t} +\alpha_{1t}W_{1t} +\alpha_{2t}\cos W_{2t} +\alpha_{3t}W_{3t}^2 +\alpha_{4t}U) \} \\
    Y_t &\sim N(\beta_{0t} + \beta_{1t}\sin W_{1t}   + \beta_{2t}W_{2t}W_{3t}   + \beta_{3t}A_t + U, 0.1)
\end{align*}
All true parameter values were generated once from a standard normal distribution. To see that parallel trends holds in this simulation, note that if $A_t$ is set to zero for all $t$, differences $Y_t-Y_{t-1}$ will not depend on $U$ because $U$ enters the model for $Y_t$ additively and does not enter the model for $W_t$. In each simulated dataset, we estimated the parameter $\mu_t=\E[Y_t(\overline A = \overline 0)]$, the intervention-specific mean under the policy to set everyone to be untreated. Nuissance functions were estimated using parametric models with treatment and outcome correctly specified (true), and misspecifying the treatment (gfal), outcome (qfal), or both (bfal); all without sample splitting. For each set of parametric nuissance functions, the estimator $\widehat \psi_t$ and variance estimator $\widehat V(\widehat \psi_t)$ were calculated. Finally, nuissance functions were estimated using Super Learner with a library including the mean, lasso, multivariate adaptive regression splines, and kernal support vector machines, using sample splitting with $M=10$ random partitions. The estimator $\widehat \psi_{t,ss}$ and variance estimator $\widehat V(\widehat \psi_{t,ss})$ were then calculated.
\par Results from the simulation study are shown in Table \ref{tbl:sim1} and Figure \ref{fig:sim1}. At all sample sizes considered (ranging from $n=1,000$ to $50,000$) and whenever either set of models is correctly specified, estimators appear to be converging at root-n rate to a normal distribution centered at zero and with variance close to the EIF-based variance estimate. This is true regardless of whether data-adaptive or parametric nuisance estimators are used, suggesting there is little loss of efficiency in using data-adaptive estimators. Though the consistency of the proposed variance estimator depends theoretically on all models being correctly specified, in this particular example the results do not appear very sensitive to this requirement. That is, the mean of estimated variances appears very close to the empirical variance of estimates over repeated simulations, even when only either the propensity score or outcome models are correct.

\begin{table}[b]
     \centering
 \begin{tabular}{r|r|r|r|r|r|r|r}
\hline
$n_{obs}$ & method & bias $(\widehat\psi_1)^2$ & $\widehat V_{sim}(\widehat\psi_1)$ & $\widehat V_{eif}(\widehat\psi_1)$ & bias $(\widehat\psi_2)^2$ & $\widehat V_{sim}(\widehat\psi_2)$ & $\widehat V_{eif}(\widehat\psi_2)$\\
\hline
1000 & bfal & 0.1269 & 0.1952 & 0.2024 & 0.2119 & 0.3131 & 0.3790\\
\hline
 & gfal & $<$0.0001 & 0.1889 & 0.1882 & 0.0001 & 0.3122 & 0.3229\\
\hline
 & qfal & $<$0.0001 & 0.1915 & 0.2138 & $<$0.0001 & 0.3233 & 0.4417\\
\hline
 & super & 0.0490 & 0.2224 & 0.1670 & 0.0029 & 0.3424 & 0.2847\\
\hline
 & true & $<$0.0001 & 0.1911 & 0.1918 & $<$0.0001 & 0.3181 & 0.3333\\
\hline
5000 & bfal & 0.1280 & 0.0349 & 0.0404 & 0.2219 & 0.0676 & 0.0759\\
\hline
 & gfal & $<$0.0001 & 0.0349 & 0.0376 & 0.0004 & 0.0629 & 0.0648\\
\hline
 & qfal & $<$0.0001 & 0.0350 & 0.0427 & 0.0002 & 0.0682 & 0.0881\\
\hline
 & super & 0.0017 & 0.0353 & 0.0347 & 0.0003 & 0.0681 & 0.0601\\
\hline
 & true & $<$0.0001 & 0.0351 & 0.0383 & 0.0004 & 0.0631 & 0.0669\\
\hline
20000 & bfal & 0.1328 & 0.0094 & 0.0101 & 0.2184 & 0.0189 & 0.0190\\
\hline
 & gfal & $<$0.0001 & 0.0094 & 0.0094 & 0.0002 & 0.0185 & 0.0162\\
\hline
 & qfal & $<$0.0001 & 0.0095 & 0.0107 & 0.0001 & 0.0190 & 0.0220\\
\hline
 & super & 0.0002 & 0.0085 & 0.0093 & $<$0.0001 & 0.0154 & 0.0162\\
\hline
 & true & $<$0.0001 & 0.0096 & 0.0096 & 0.0002 & 0.0187 & 0.0167\\
\hline
50000 & bfal & 0.1310 & 0.0046 & 0.0040 & 0.2158 & 0.0061 & 0.0076\\
\hline
 & gfal & $<$0.0001 & 0.0047 & 0.0038 & 0.0001 & 0.0059 & 0.0065\\
\hline
 & qfal & $<$0.0001 & 0.0046 & 0.0043 & $<$0.0001 & 0.0060 & 0.0088\\
\hline
 & super & 0.0001 & 0.0042 & 0.0038 & 0.0002 & 0.0060 & 0.0066\\
\hline
 & true & $<$0.0001 & 0.0046 & 0.0038 & 0.0001 & 0.0058 & 0.0067\\
\hline
\end{tabular}
    \caption{Results of simulation study, with $n_{sims}=500$, and $M=2$ (bias and variance numbers multiplied by 100). $\widehat V_{sim}$ is the empirical variance across simulation runs, and $\widehat V_{eif}$ is the mean of influence curve-based variance estimates. }
    \label{tbl:sim1}
\end{table}

\begin{figure}[b]
    \centering
    \includegraphics[scale=.8]{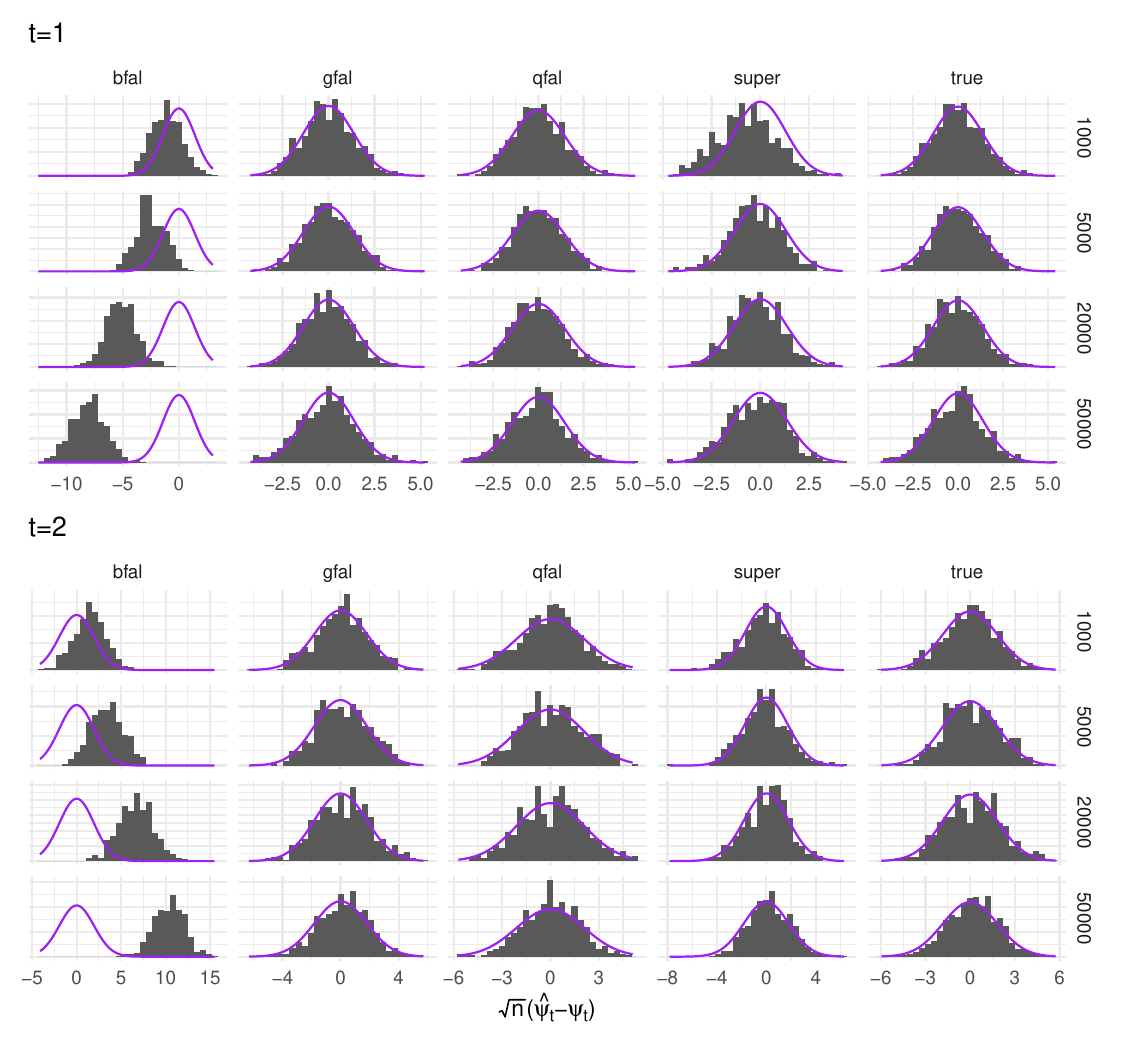}
    \caption{Simulation results. Histograms are error estimates scaled by $\sqrt{n}$. Curves are normal density functions with mean 0 and variance equal to the mean of estimated efficient influence function-based variances across simulation runs.}
    \label{fig:sim1}
\end{figure}

\section{Minimum wages and self-rated health}

In this section, we apply our proposed approach to estimate the effect of recent minimum wage increases on self-reported health using data from the Behavioral Risk Factor Surveillence Survey (BRFSS). BRFSS is a cross-sectional telephone survey focused conducted among more than 400,000 adults annually \citep{silva2014behavioral}. BRFSS data are supplemented using annual state-level minimum wage data and control variables related to state economic conditions from the U.S. Department of Labor, Bureau of Labor Statistics, and Bureau of Economic Analysis \citep{us2023changes}. These data were originally compiled by \citet{sigaud2022effects} for use in a similar difference-in-difference based analysis. 

Our target estimand is the expected self-reported health for the US population, had the US federal government counterfactually implemented a minimum wage on par with that of New York State (NYS). Specifically, consider a Likert scale measure of self-reported health ranging from 1-5 corresponding to to self-rated ``excellent", ``very good", ``good", ``fair", and ``poor" general health. Let $Y_{st}$ denote the proportion of residents in a given state, $s$, who report having ``excellent" or ``very good" health during year $t$. Let $A_{st}$ be an indicator of whether state $s$ has minimum wage greater than or equal to that of NYS in year $t$. Then the target estimand is the intervention-specific mean $\mu_t = E[Y_{st}(\bar{a}^*_t = \bar{1})]$: the expected proportion of individuals who would report excellent or very good health if the US had a federal minimum wage law consistent with that of NYS. If a given state has implemented a minimum wage that is greater than that of NYS, then this state is consistent the hypothetical federal minimum wage and is not intervened upon. Importantly, because the exposure for each state in a given year explicitly depends upon the minimum wage in NYS, data can not be considered as independent and identically distributed, which is an issue because the proposed inferential techniques rely on this assumption. Instead, we draw all conclusions conditional on knowledge of NYS' minimum wage trajectory over the period of study. Conditional on this information, data may be considered IID and inference can proceed normally.  

Our analysis is similar in scope to that of \citet{sigaud2022effects}. In their analysis, the authors estimate the following two-way fixed effects model:
\begin{equation*}
 E[Y_{istm}] = \beta_0 + \beta_1 MW_{s,t-1} + \beta_2 Z_{itm} + \beta_3 X_{st} + State_s + Year_t + Month_m + \phi_{st} + \epsilon_{istm} 
\end{equation*}
Where $MW_{st}$ denotes the minimum wage in state $s$ during year $t-1$, $Z_{itm}$ denotes individual-level control variables, and $X_{st}$ denotes state-level controls.  Note that our approach differs from \citet{sigaud2022effects} in three important ways: (1) the estimand of interest, (2) the identification and modeling assumptions underlying our approach, and (3) the framing of the minimum wage exposure. While our estimator targets the intervention-specific mean, the target estimand of the TWFE model is a weighted-average of all possible 2$\times$2 DID estimators \citep{goodman2021difference}, and is thus less immediately interpretable. Moreoever, the TWFE approach relies on additional assumptions for identification; in particular, that adjustment covariates $X_{st}$ are not affected by previous treatment. This may not be plausible: state-level control covariates such as  unemployment rate and per capita income are very likely affected by previous minimum wage changes. 

We incorporate data from the 50 U.S. states and Washington, D.C.. Between 2013 and 2019, $n = 3,111,090$ participants were surveyed. State-level confounding variables considered in our analysis include economic indices such as the maximum Temporary Assistance for Needy Families (TANF) benefit for a family of four, the maximum Supplemental Nutrition Assistance Program (SNAP) benefit for a family of four, the state Earned Income Tax Credit (EITC), the unemployment rate, and the per capita personal income as well as state-level demographic summaries of sex (proportion male), educational attainment (proportion graduated high school, attended college or technical school, and graduated from college or technical school), race (proportion non-hispanic white, non-hispanic black, and other), and age (proportion older than 64). State-level summaries were obtained by averaging the individual-level BRFSS data using survey design weighting. Figure \ref{fig:minimum-wage} presents the minimum wages for the 50 states and Washington DC over the period from 2013 to 2019. In 2013, when the NYS minimum wage was set to \$7.25, all states had minimum wages at least as high as NYS. By 2019, only 5 states (CA, DC, MA, VT, and WA) had maintained minimum wages at least as high as NYS throughout the study period. 

Estimating the target parameter required first estimating nuisance functions: (1) the outcome models $Q^{j,k,m}$ for $j = 2019, \ldots, 2014$, $k = j \text{ or } j+1$, and $m = 2013, \ldots,j$ and (2) the exposure propensity model $g_m$ for $m = 2013, \ldots, 2019$. Outcome models were estimated using SuperLearner, an ensemble learning approach which models the outcome as an optimally weighted combination of predictions from a set of candidate learners. Our candidate learners included variations of random forests and multivariate adaptive regression. For each type of learner for the outcome in year $j$, we incorporated a version which used all observed covariates up to year $j$, as well as versions which screened covariates based on having a significant correlation with the selected outcome or based on being selected from a LASSO regression. To avoid Donsker class restrictions, we used sample splitting based on $K = 2$ splits of the data. The cumulative propensity model, $g_m$, was constructed using the product of time-specific treatment propensities, $f(A_k = 1|\bar{W}_k, \bar{A}_{k-1} = \bar{1})$. These propensity models were also fit using SuperLearner, with candidate learners including random forests and elastic net regression. Note that treatment propensities only needed to be modelled at times when the number of ``consistent" states changed. For example, from 2013 to 2014, all 51 states were consistent and no estimation of propensity was required ($f(A_{2014} = 1|\overline{W}_{2014}, A_{2013} = 1) = 1$). Models for treatment propensity had to be fit for the years $2015,\ 2016,\ 2018$, and $2019$. Because the propensity score model conditions on having remained consistent with NYS up to year $t-1$, small sample sizes in the final years led to difficulties in estimating these models. For example, 8 states were consistent with NYS minimum wage in 2018 and, of these states, 6 remained consistent in 2019. In particular, when taking into account the sample splitting required for both avoiding Donsker class conditions and inherent in the optimization of SuperLearner model parameters were occasionally unidentifiable. To address this concern, we assume the same propensity score model for all years (i.e. $g_m = g$). In this simplified model, exposure in the following year is predicted based on observed covariates over the previous two years.

\begin{table}[b]
    \centering
    \small\begin{tabular}{cccccccccccccccc}
    \toprule
     & & \multicolumn{4}{c}{\underline{Observed}} & & \multicolumn{4}{c}{\underline{Counterfactual}} & & \multicolumn{4}{c}{\underline{Difference}} \\
     Year    & & Est. & SE & Lower & Upper & & Est. & SE & Lower & Upper & & Est. & SE & Lower & Upper \\
    \midrule
    \midrule
     2015    & & 52.0\% & 0.66 & 50.7\% & 53.3\% & & 51.6\% & 0.78 & 50.1\% & 53.2\% & & -0.40 & 0.35 & -1.09  & 0.28  \\
     2016    & & 51.5\% & 0.64 & 50.2\% & 52.8\% & & 51.7\% & 0.82 & 50.0\% & 53.3\% & & 0.16 & 0.44 & -0.70 & 1.01  \\
     2017    & & 50.3\% & 0.67 & 49.0\% & 51.6\% & & 50.7\% & 0.84 & 49.1\% & 52.4\% & & 0.42 & 0.50 & -0.55 & 1.40  \\
     2018    & & 50.5\% & 0.62 & 49.3\% & 51.7\% & & 50.5\% & 0.90 & 48.8\% & 52.3\% & & 0.01 & 0.56 & -1.08 & 1.10 \\
     2019    & & 49.7\% & 0.62 & 48.5\% & 50.9\% & & 50.4\% & 0.86 & 48.8\% & 52.1\% & & 0.75 & 0.64 & -0.51 & 2.00  \\
    \bottomrule
    \end{tabular}
    \caption{Estimated prevalence of ``excellent'' or ``very good'' self-rated health in the US, based on data from the 2013-2019 Behavioral Risk Factor Surveillance Study. Counterfactual estimates are based on our proposed estimators implemented with Super Learner.}
    \label{tab:data_results}
\end{table}

\begin{figure}[b]
    \centering
    \includegraphics[scale = 0.5]{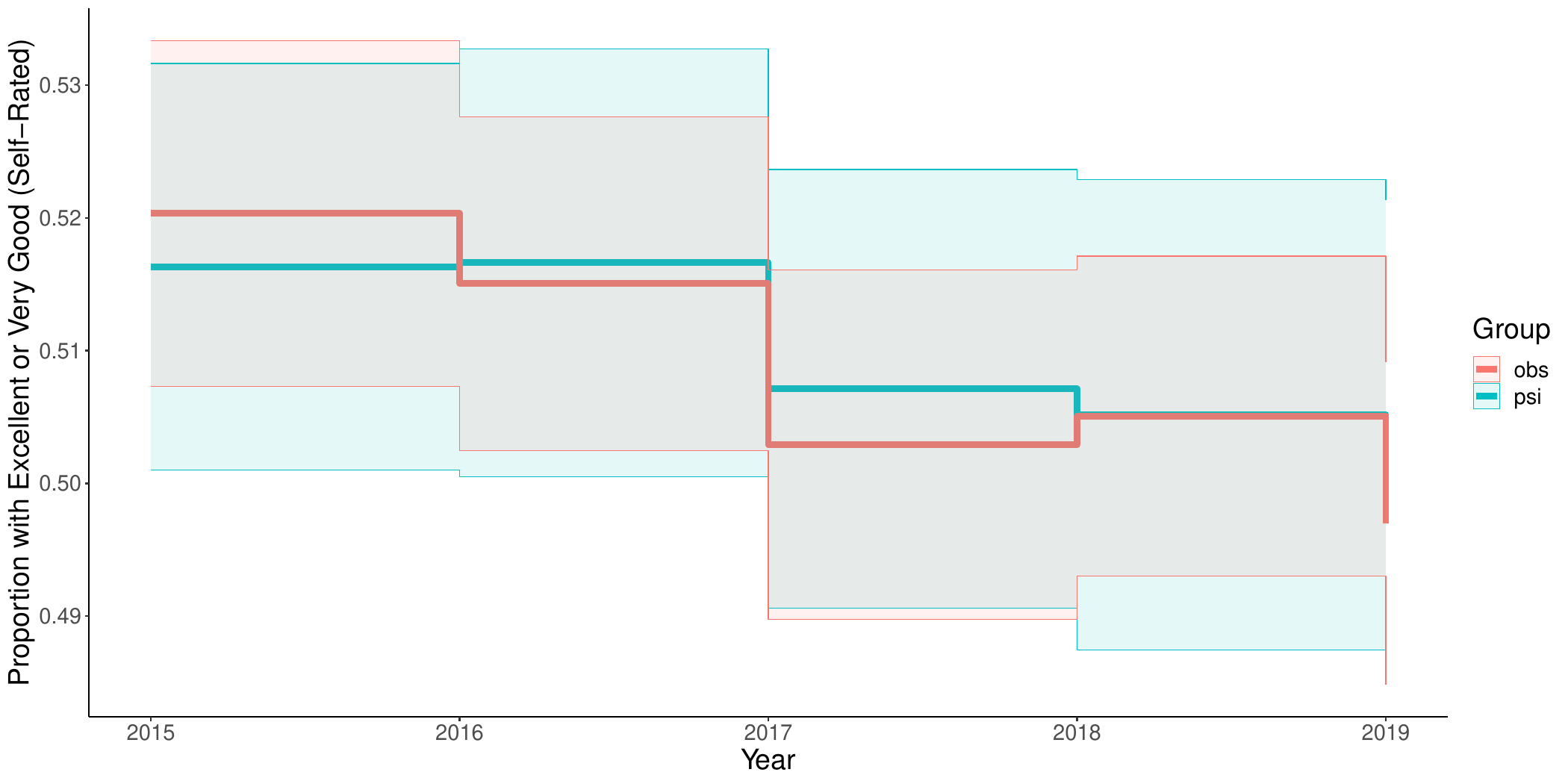}
    \caption{Estimated proportion of participants reporting "Excellent" or "Very Good" Health between 2015 and 2019. Proportions provided for the observed population and for a hypothetical population of states with minimum wage laws consistent with NYS. 95\% point-wise confidence intervals included.}
    \label{fig:evo2}
\end{figure}

The primary results of this analysis can be found in Table \ref{tab:data_results} and Figure \ref{fig:evo2}. Table \ref{tab:data_results} presents the estimated proportion (with associated standard errors) of individuals reporting ``excellent" or ``very good" among (1) the observed population, and (2) a hypothetical population in which the federal government had a minimum wage law consistent with that of NYS for each year between 2015 and 2019. Also provided are the estimated differences in proportions between the two populations under each year of study. Figure \ref{fig:evo2} presents a step function visualizing these results. Importantly, we do not observe sufficient evidence to claim that setting the federal minimum wage to that of NYS would have impacted the self-rated health of individuals during the years under consideration. The largest estimated difference in proportion of individuals with high-levels of self-reported health occured in the year 2019. In this year, the observed proportion of participants who reporting high-levels of health was 49.7\% (95\% CI: 48.5\%, 50.9 \%).

\begin{figure}[b]
    \centering
    \includegraphics[scale = 0.5]{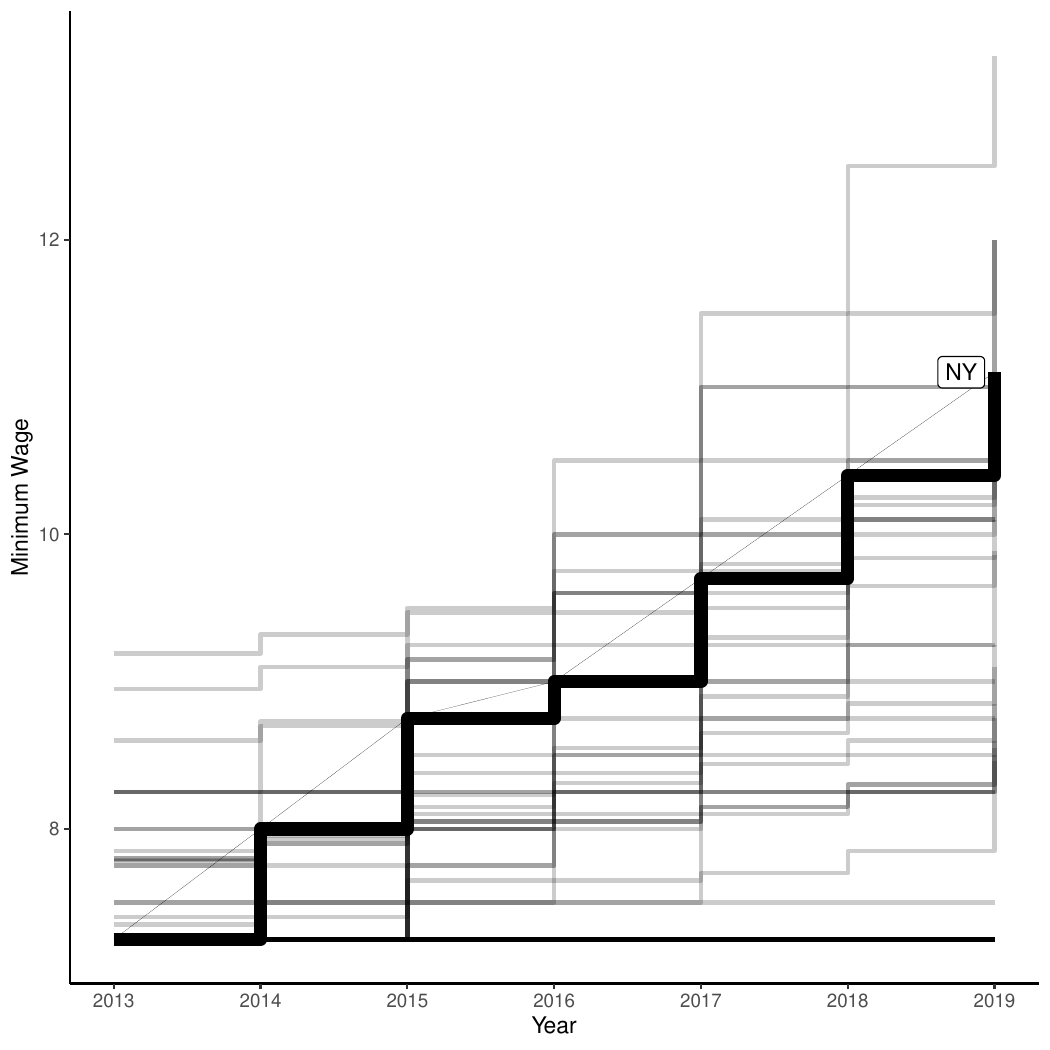}
    \caption{State-level minimum wage levels (\$USD) between 2013 and 2019. New York State minimum wage highlighted in bold.}
    \label{fig:minimum-wage}
\end{figure}



\section{Discussion}
This paper proposed novel machine learning-based estimators for longitudinal DID settings with time-varying treatments and/or covariates, focused on the intervention-specific mean. By being based on the efficient influence function, the proposed estimators are efficient and allow appropriate inference even when the nuisance functions are estimated using complex machine learning algorithms. In our simulations, the estimators achieved similar finite-sample precision when using machine learning as when correctly specified parametric models were used. Investigators rarely have knowledge about the functional form of nuisance models, especially when these are higher dimensional; thus interest in data adaptive estimators has been growing \citep{diaz2023nonparametric,chernozhukov2018double, chang2020double, shahn2022structural}. Our focus on the intervention-specific mean stems from the fact that these parameters can sometimes be highly policy relevant. For example, \citet{leifheit2021expiring} estimated counterfactual COVID-19 incidence rates if US-state-level eviction moratoria had not expired; these results were then cited in the CDC order calling for a federal eviction moratoria \citep{centers2021temporary}. We also are among the first DID papers to accomodate time-varying covariates; particularly in the likely scenario that these may be affected by prior treatment. For example, \citet{callaway2023policy} discuss economic outcomes of pandemic-related policies. These policy decisions are based on the current infection rates, which are in turn both influenced by prior policy and influence future economic outcomes. Such treatment-confounder, which we conjecture is operating in many contexts where DID is used, would invalidate most existing DID methods but is adequately addressed in our approach, so long as it is consistent with the parallel trends assumption (the simulation setup in Section \ref{sec:sim} provides such an examples). 
\par In contrast to exchangeability assumptions, the parallel trends assumption is semiparametric and cannot be fully justified using causal diagrams; thus it can be difficult to justify a priori. An active area of research aims to describe, outside of any parametric model, what conditional independencies and/or functional form restrictions in a structural model imply and/or are implied by parallel trends \citep{ghanem2022selection, marx2024parallel}. In parametric models, parallel trends generally requires that any unmeasured confounders are time-invariant in both their distribution and their association with outcomes over time. Thus, we may reasonably expect that a semiparametric structural model consistent with parallel trends would need to be partially linear with respect to the unmeasured covariates, though there is more to be explored here. It is important to note the scale-dependence of parallel trends; the version presented in this paper is linear (as typical in DID), but the assumption could be expressed on a nonlinear scale as well. A linear scale may not be reasonable for many kinds of outcomes, such as survival times or other distributions expected to be skewed. 
\par Regarding estimation in practice, choosing the machine learning library and/or the number of partitions in cross-fitting may be challenging. In general, selecting a larger number of partitions reduces bias by ensuring nuisance models are fit on a larger portion of the data, but can increase computation time.  Though our approach using a one-step estimator and (single-) cross-fitting approach guarantees root-n consistency under the stated assumptions, finite sample bias may be further reduced by applying other approaches. Recent work by \citet{zivich2021machine} describe a ``double cross-fitting'' approach that eliminates second-order bias \citep{newey2018cross}, wherein the data are repeatedly partitioned into three components: (1) a training set for the outcome regression models, (2) a training set for the exposure propensity model, and (3) an treatment effect evaluation set. To our knowledge, this approach has not been extended to time-varying settings such as ours. It may also be useful to develop target maximum likelihood estimators \citet{van2006targeted, van2012targeted} for this setting, as the proposed approach does not yield a substitution estimator and so could produce estimates outside the known parameter range (e.g. risks outside of $[0,1]$), possibly contributing to finite sample bias.

\printbibliography

\appendix

\section{Proof of Theorem 1}
By the linearity property of influence functions, it can be seen that:
\begin{equation*}
    IF(\psi_t) = IF(\E[Y_0]) + \sum_{k=1}^t (IF(\phi_{k,k}) - IF(\phi_{k-1,k})).
\end{equation*}
Finding the influence function for $\psi_t$ simplifies to identifying the influence function for $\phi_{j,k}$. By \citet{van2012targeted} this is given by:
\begin{equation*}
    IF(\phi_{j,k}) = \sum_{m=1}^k \frac{I(\overline{A}_m = \overline{a}_m^*)}{g_m(\overline W_m)}(Q^{j,k,m+1}-Q^{j,k,m}) + Q^{j,k,1} - \phi_{j,k} 
\end{equation*}
To see why this is true, consider the following:
\begin{align*}
    IF(\phi_{j,k}) &= IF(E[Q^{j,k,1}]) \\
    &= IF\bigg(\int Q^{j,k,1} dF(\overline w_1) \bigg) \\
    &= \int IF(Q^{j,k,1})dF(\overline w_1) + \int Q^{j,k,1} IF(dF(\overline w_1))
\end{align*}
Focus attention on the second term:
\begin{align*}
\int Q^{j,k,1} IF(dF(\overline w_1)) &= \int Q^{j,k,1}IF(f(\overline w_1))d\overline w_1 \\
&= \int Q^{j,k,1}(I(\overline W_1 = \overline w_1) - f(\overline w_1))d\overline w_1 \\
&= \int Q^{j,k,1}I(\overline W_1 = \overline w_1)d\overline w_1 - \int Q^{j,k,1}f(\overline w_1)d\overline w_1 \\
&= Q^{j,k,1} - E[Q^{j,k,1}] \\
&= Q^{j,k,1} - \phi_{j,k}.
\end{align*}
Next, we will show that for $m = 0,...,k$ the following relationship holds:
\small\begin{align}
    \int IF(Q^{j,k,m}) \prod_{s=0}^m dF(w_s|\overline{W}_{s-1} =  \overline w_{s-1}, \overline{A}_{s-1} = \overline{a}_{s-1}^*) &= \int IF(Q^{j,k,m+1}) \prod_{s=0}^{m+1} dF(w_s|\overline{W}_{s-1} = \overline w_{s-1}, \overline{A}_{s-1} = \overline{a}_{s-1}^*) \nonumber \\ 
    &\qquad + \frac{I(\overline{A}_m = \overline{a}_m^*)}{g_m(\overline w_m)}(Q^{j,k,m+1} - Q^{j,k,m}) 
\end{align}
Consider the following:
\begin{align*}
    \int IF(Q^{j,k,m}) &\prod_{s=0}^m dF(w_s|\overline{W}_{s-1} =  w_{s-1}, \overline{A}_{s-1} = \overline{a}_{s-1}^*) \\
    &= \int IF\bigg\{ \int Q^{j,k,m+1} dF(w_{m+1}|\overline{w}_{m}, \overline{a}^*_{m}) \bigg\} \prod_{s=0}^m dF(w_s|\overline{W}_{s-1} =  w_{s-1}, \overline{A}_{s-1} = \overline{a}_{s-1}^*) \\
    &= \int \int IF\bigg\{  Q^{j,k,m+1}\bigg\} \prod_{s=0}^{m+1} dF(w_s|\overline{W}_{s-1} =  w_{s-1}, \overline{A}_{s-1} = \overline{a}_{s-1}^*) \\
    &\qquad + \int  \int Q^{j,k,m+1} IF\bigg\{dF(w_{m+1}|\overline{w}_{m}, \overline{a}^*_{m}) \bigg\} \prod_{s=0}^m dF(w_s|\overline{W}_{s-1} =  w_{s-1}, \overline{A}_{s-1} = \overline{a}_{s-1}^*)
\end{align*}
\normalsize Focus on the second term:
\tiny\begin{align*}
    \int  \int Q^{j,k,m+1} &IF\bigg\{dF(w_{m+1}|\overline{w}_{m}, \overline{a}^*_{m}) \bigg\} \prod_{s=0}^m dF(w_s|\overline{W}_{s-1} =  w_{s-1}, \overline{A}_{s-1} = \overline{a}_{s-1}^*) \\
    &= \int  \int Q^{j,k,m+1} \bigg\{\frac{I{(\overline{A}_{m} = \overline{a}_{m}^*, \overline{W}_m = \overline{w}_m)}}{g_m(\overline w_m)\prod_{l=0}^{m}f(w_l|\overline{w}_{l-1},\overline{a}_{l-1}^*)} (I(W_{m+1} = w_{m+1}) - f(w_{m+1}|\overline{w}_{m},\overline{a}_{m}^*))dw_{m+1}\bigg\} \\
    &\qquad \times \prod_{s=0}^m dF(w_s|\overline{W}_{s-1} =  w_{s-1}, \overline{A}_{s-1} = \overline{a}_{s-1}^*) \\
    &= \int  \int Q^{j,k,m+1} \bigg\{\frac{I{(\overline{A}_{m} = \overline{a}_{m}^*, \overline{W}_m = \overline{w}_m)}}{g_m(\overline w_m)\prod_{l=0}^{m}f(w_l|\overline{w}_{l-1},\overline{a}_{l-1}^*)}I(W_{m+1} = w_{m+1})dw_{m+1}\bigg\}\prod_{s=0}^m dF(w_s|\overline{W}_{s-1} =  w_{s-1}, \overline{A}_{s-1} = \overline{a}_{s-1}^*) \\
    &\qquad - \int  \int Q^{j,k,m+1} \bigg\{\frac{I{(\overline{A}_{m} = \overline{a}_{m}^*, \overline{W}_m = \overline{w}_m)}}{g_m(\overline w_m)\prod_{l=0}^{m}f(w_l|\overline{w}_{l-1},\overline{a}_{l-1}^*)} f(w_{m+1}|\overline{w}_{m},\overline{a}_{m}^*)dw_{m+1}\bigg\}\prod_{s=0}^m dF(w_s|\overline{W}_{s-1} =  w_{s-1}, \overline{A}_{s-1} = \overline{a}_{s-1}^*) \\
    &= \int \frac{I{(\overline{A}_{m} = \overline{a}_{m}^*, \overline{W}_m = \overline{w}_m)}}{g_m(\overline w_m)\prod_{l=0}^{m}f(w_l|\overline{w}_{l-1},\overline{a}_{l-1}^*)}\bigg\{\int Q^{j,k,m+1}I(W_{m+1} = w_{m+1})dw_{m+1}\bigg\}\prod_{s=0}^m dF(w_s|\overline{W}_{s-1} =  w_{s-1}, \overline{A}_{s-1} = \overline{a}_{s-1}^*) \\
    &\qquad - \int \frac{I{(\overline{A}_{m} = \overline{a}_{m}^*, \overline{W}_m = \overline{w}_m)}}{g_m(\overline w_m)\prod_{l=0}^{m}f(w_l|\overline{w}_{l-1},\overline{a}_{l-1}^*)} \bigg\{\int Q^{j,k,m+1} f(w_{m+1}|\overline{w}_{m},\overline{a}_{m}^*)dw_{m+1}\bigg\}\prod_{s=0}^m dF(w_s|\overline{W}_{s-1} =  w_{s-1}, \overline{A}_{s-1} = \overline{a}_{s-1}^*) \\
    &= \int \frac{I{(\overline{A}_{m} = \overline{a}_{m}^*, \overline{W}_m = \overline{w}_m)}}{g_m(\overline w_m)\prod_{l=0}^{m}f(w_l|\overline{w}_{l-1},\overline{a}_{l-1}^*)}Q^{j,k,m+1}(\overline{a}^*, \overline{W}_{m+1} = (W_{m+1}, \overline{w}_m))\prod_{s=0}^m dF(w_s|\overline{W}_{s-1} =  w_{s-1}, \overline{A}_{s-1} = \overline{a}_{s-1}^*) \\
    &\qquad - \int \frac{I{(\overline{A}_{m} = \overline{a}_{m}^*, \overline{W}_m = \overline{w}_m)}}{g_m(\overline w_m)\prod_{l=0}^{m}f(w_l|\overline{w}_{l-1},\overline{a}_{l-1}^*)} Q^{j,k,m}(\overline{a}^*, \overline{w}_m)\prod_{s=0}^m dF(w_s|\overline{W}_{s-1} =  w_{s-1}, \overline{A}_{s-1} = \overline{a}_{s-1}^*) \\
    &= \frac{I(\overline{A}_{m} = \overline{a}_{m}^*)}{g_m(\overline{W}_m)}(Q^{j,k,m+1}  - Q^{j,k,m} )
\end{align*}
\normalsize
This implies the overall influence function because:
\begin{align*}
IF(\phi_{j,k}) &= \int IF(Q^{j,k,1})dF(\overline w_1) + \int Q^{j,k,1} IF(dF(\overline w_1)) \\
&= \int IF(Q^{j,k,k+1}) \prod_{s=0}^{k+1} dF(w_s|\overline{W}_{s-1} = w_{s-1}, \overline{A}_{s-1} = \overline{a}_{s-1}^*) \\
&\qquad + \sum_{m=1}^k \frac{I(\overline{A}_m = \overline{a}_m^*)}{g_m(\overline W_m)}(Q^{j,k,m+1} - Q^{j,k,m}) \\
&\qquad + Q^{j,k,1} - \phi_{j,k} \\
&= \sum_{m=1}^k \frac{I(\overline{A}_m = \overline{a}_m^*)}{g_m(\overline W_m)}(Q^{j,k,m+1} - Q^{j,k,m}) + Q^{j,k,1} - \phi_{j,k}
\end{align*}
Where the last equality follows from the fact that $IF(Q^{j,k,k+1}) = IF(Y_j) = 0$

\section{Proof of Theorem 2}

    By Lemma 1 of \citet{luedtke2017sequential} we have 
    \[\widetilde{\phi}_{jk} - \phi_{jk} = (P_n - P)IF(\phi_{jk}) + (P_n-P)(\widehat{IF}(\phi_{jk}) - IF(\phi_{jk})) + Rem_{jk},\]
    where $\widehat{IF}(\phi_{jk})$ denotes the IF of Theorem 1 with $\Theta$ replaced by $\widehat\Theta$, and where 
    \[Rem_{jk} = \sum_{m=1}^k \int c_m(\overline w_m) \{g_m(\overline w_m) - \widehat{g}_m(\overline w_m)\}\{Q^{j,k,m}(\overline w_m) - \widehat{Q}^{j,k,m}(\overline w_m)\}dP(\overline w_m),\]
    where $c_m(\overline W_m)$ is a constant given in the original Lemma, which is bounded in probability according to C2. Under C1, Lemma 19.24 of \citet{van2000asymptotic} shows that the term $(P_n-P)(\widehat{IF}(\phi_{jk}) - IF(\phi_{jk}))$ is $o_P(n^{-1/2})$. Likewise, using the Cauchy-Schwarz inequality we have
    \[|Rem_{jk}| \leq \norm{g_m - \widehat{g}_m}\norm{Q^{j,k,m} - \widehat{Q}^{j,k,m}}.\]
    Applying the above to 
   \[\psi_t=\E[Y_0] + \sum_{k=1}^t (\phi_{k,k} - \phi_{k-1, k} )
     ;\quad \widetilde\psi_t=P_n Y_0 + \sum_{k=1}^t (\widetilde\phi_{k,k} - \widetilde\phi_{k-1, k} )
     \]
     together with Theorem 1 yields the result. 

\end{document}